\documentclass[10pt,twocolumn,twoside]{IEEEtran}
\IEEEoverridecommandlockouts

\usepackage{cite}
\usepackage{amsmath,amssymb,amsfonts}
\usepackage{graphicx}
\usepackage{textcomp}
\usepackage{balance}
\usepackage{xcolor}
\usepackage[utf8]{inputenc}

\usepackage{enumitem}
\usepackage{amsthm}
\usepackage{accents}
\newcommand{\ubar}[1]{\underaccent{\bar}{#1}}
\newlist{abbrv}{itemize}{1}
\setlist[abbrv,1]{label=,labelwidth=0.5in,align=parleft,itemsep=0.1\baselineskip,leftmargin=!}
\usepackage{float}
\floatstyle{ruled}
\newfloat{model}{thp}{lop}
\floatname{model}{Model}

\newtheorem{theorem}{Theorem}

\usepackage{todonotes}

\usepackage{algorithm}
\usepackage{algorithmicx}
\usepackage{algpseudocode}

\def\BibTeX{{\rm B\kern-.05em{\sc i\kern-.025em b}\kern-.08em
    T\kern-.1667em\lower.7ex\hbox{E}\kern-.125emX}}
\usepackage{lipsum}

\title{\LARGE \bf Trust-Region Benders for Scalable Multiperiod Expansion Planning for Transmission and Storage}

\author{Kevin Wu$^\dagger$, Rabab Haider$^\ddagger$, Pascal Van Hentenryck$^\dagger$ 
\vspace{-0.5cm}
 \thanks{$^\dagger$Kevin Wu and Pascal Van Hentenryck are with the H. Milton Stewart School of Industrial and Systems Engineering, Georgia Institute of Technology, Atlanta, GA, United States. 
Email: \{kwu381, pvh\}@gatech.edu}
\thanks{$^\ddagger$Rabab Haider is with the Civil and Environmental Engineering, University of Michigan, Ann Arbor, MI, United States. 
Email: rababh@umich.edu}
}

\begin{document}
\begingroup
\allowdisplaybreaks

\maketitle
\thispagestyle{plain}

\begin{abstract}
Transmission expansion planning (TEP) seeks to identify cost-effective investment strategies that enable power systems to reliably meet future demand, while accounting for the location, timing, and operational impacts of infrastructure upgrades. As demand grows and renewable penetration increases, there is growing interest in jointly modeling energy storage investments to enhance operational flexibility. This work presents a multiperiod formulation that co-optimizes transmission reconductoring and energy storage installation across multiple planning periods. To solve the resulting large-scale problem efficiently, a Benders decomposition with per-period trust-region stabilization and a relaxation-based warm-start procedure is proposed. The approach is demonstrated on a 2,000-bus Texas system, producing an investment plan spanning 2030–2045 with optimality gap below 1.0\%, and is compared against a rolling-horizon approach to highlight the benefits of multiperiod coordination.
\end{abstract}



\newcommand{\db}{\delta B}
\newcommand{\dbhi}{\bar{\delta B}}
\newcommand{\dblo}{\ubar{\delta B}}
\newcommand{\dva}{{\Delta \theta}}
\newcommand{\dvahi}{\bar{\Delta \theta}}
\newcommand{\dvalo}{\ubar{\Delta \theta}}

\newcommand{\LB}{\mathrm{LB}}
\newcommand{\UB}{\mathrm{UB}}
\newcommand{\MP}{\mathrm{MP}}
\newcommand{\SP}{\mathrm{SP}}
\newcommand{\op}{\mathrm{op}}
\newcommand{\rad}{\mathrm{rad}}

\section*{Nomenclature}
\label{sec:nomenclature}
\paragraph{Sets}

\newcommand{\R}{\mathcal{R}}
\newcommand{\N}{\mathcal{N}}
\newcommand{\T}{\mathcal{T}}
\newcommand{\E}{\mathcal{E}}
\newcommand{\Y}{\mathcal{Y}}

\newcommand{\pd}{\mathbf{p}^{\text{d}}}
\newcommand{\pgmin}{\mathbf{\ubar{p}}^{\text{g}}}
\newcommand{\pgmax}{\mathbf{\bar{p}}^{\text{g}}}
\newcommand{\pfmax}{\mathbf{\bar{p}}^{\text{f}}}
\newcommand{\dvamin}{\ubar{\theta}}
\newcommand{\dvamax}{\bar{\theta}}

\newcommand{\df}{\delta \pf}
\newcommand{\dX}{\delta X}
\newcommand{\dXmin}{\ubar{\delta X}}
\newcommand{\dXmax}{\bar{\delta X}}
\newcommand{\dB}{\delta B}
\newcommand{\dBmin}{\ubar{\delta B}}
\newcommand{\dBmax}{\bar{\delta B}}

\newcommand{\pg}{\mathbf{p}^{\text{g}}}
\newcommand{\pf}{\mathbf{p}^{\text{f}}}
\newcommand{\va}{\theta}

\begin{abbrv}
    \item[$\N$] Set of buses; $\N = \{1, ..., N\}$; indexed by $i$
    \item[$\E$] Set of branches; $\E = \{1, ..., E\}$; indexed by $ij$
    \item[$\R$] Set of days; $\R = \{1, ..., R\}$; indexed by $r$
    \item[$\Y$] Set of periods; $\Y = \{1, ..., Y\}$; indexed by $y$
    \item[$\T$] Set of hours; $\T = \{1, ..., T\}$; indexed by $t$
\end{abbrv}
\noindent{Unless otherwise specified, all parameters and operational variables are indexed by $(y,r,t)$, denoting planning period $y$, representative day $r$, and hour $t$.}
\vspace{\baselineskip} 

\paragraph{Parameters}

\begin{abbrv}
    \item[$\pd_{i,y,r,t}$] Load at node $i$
    \item[$\pgmin_{i,y,r,t}$] min output of generator $i$
    \item[$\pgmax_{i,y,r,t}$] max output of generator $i$
    \item[$c^{\text{rec}}_{ij}$] Annualized reconductoring cost of edge $ij$
    \item[$c^{\text{stor}}_{i}$] Annualized storage installation cost at node $i$
    \item[$c^{\text{G}}_{i}(\cdot)$] Production cost function of generator $i$
    \item[$c^{\text{S}}_{i}(\cdot)$] Storage charge/discharge cost of node $i$
    \item[$\Omega_{\R}(\cdot)$] Probability mass function for set of days $\R$
    \item[$\lambda$] Load shedding penalty cost
    \item[$\Delta^c_{ij}$] Reconductoring capacity increment of branch $ij$
    \item[$\pfmax_{ij}$] Thermal limit of branch $ij$
    \item[$\alpha$] Interest rate
\end{abbrv}

\paragraph{Variables}

\begin{abbrv}
    \item[$\gamma_{ij, y}$] Reconductoring of branch $ij$ by year $y$
    \item[$\sigma_{i, y}$] Installed storage energy rating at node $i$ by year $y$
    \item[$\text{soc}_{i,y,r,t}$] Storage SoC at node $i$ 
    \item[$\text{ch}_{i,y,r,t}$] Storage charge at node $i$ 
    \item[$\text{dis}_{i,y,r,t}$] Storage discharge at node $i$ 
    \item[$\pg_{i,y,r,t}$] Output of generator $i$ 
    \item[$\pf_{ij,y,r,t}$] Power flow on branch $ij$ 
    \item[$\xi_{i,y,r,t}$] Load shed at node $i$ 
\end{abbrv}

\section{Introduction}
\label{sec:introduction}
Power grids in the U.S. are rapidly decarbonizing to address climate challenges while expanding to meet increasing demand driven by electrification, manufacturing, and computing. Infrastructure expansion is required to ensure the reliable delivery of power to load centers and sufficient operational flexibility to accommodate diverse load patterns and intermittent renewable generation, posing unprecedented challenges for system planners tasked with solving transmission expansion planning (TEP) problems.

TEP is typically formulated as a mixed-integer program (MIP) and is computationally challenging to solve due to discrete investment decisions and the need to capture high spatial, temporal, and operational resolutions. Given the need for operational flexibility \cite{levin2023energy_nature_review}, there is growing interest in TEP formulations that also model storage investment and operation. The TEP+Storage model jointly considers multiple planning periods and coupled hourly operational decisions, where the location, timing, and operation of investments must be co-optimized, dramatically increasing problem scale. In this paper, the term \textit{multiperiod} refers to planning decisions across multiple investment planning periods.

TEP+Storage formulations differ with respect to investment timing, network representation, temporal resolution, and decomposition strategy. The discussion below focuses on the subset of studies most closely related to the present work, namely multiperiod planning formulations with explicit investment timing decisions and high-resolution network representations.

Existing multiperiod transmission and capacity expansion planning models are often evaluated on relatively low-dimensional test systems (e.g., $<30$ regions). For example, \cite{ramirez2020robust_multiperiod_benders} introduces a multiperiod Benders decomposition framework but does not explicitly model storage, while \cite{heuberger2017power_tech_cost_learning} and \cite{goke2024stabilized} incorporate storage investment but consider a single-node system and a four-region system, respectively. Similarly, \cite{pecci2025regularized} propose and evaluate a regularized Benders approach for multiperiod capacity expansion planning problem on a 26-zone contiguous US model. These studies focus on algorithmic development but do not test the viability of these approaches to large-scale networks. The aggregated spatial representations cannot capture high-resolution network effects such as congestion patterns and storage–transmission interactions.

Several recent studies addressing high spatial resolution TEP+Storage employ decomposition techniques such as Benders decomposition or progressive hedging, but are limited to a single planning period \cite{piansky2025optimizing_cats_wildfire_benders, zuluaga2024cats_progressive_hedging_ph}, or rely on myopic rolling-horizon approximations for multiperiod investment planning \cite{wu2025high}. In contrast, large-scale multiperiod expansion models typically adopt aggregated or transport-based transmission representations that relax network physics \cite{ReEDS}. While such approaches improve tractability, they impact decisions on siting, costs, and operations. Models with low spatial resolution lack the locational accuracy to sufficiently inform decision making \cite{Jacobson, XuHobbs}.

This work addresses these challenges by extending the trust-region Benders (TR-Benders) decomposition framework of Wu et al. \cite{wu2025high} to a multiperiod transmission and energy storage expansion planning problem with intertemporal investment coordination. The proposed formulation jointly optimizes transmission and storage investments across multiple planning periods while preserving DC power flow at hourly operational resolution. To efficiently solve the resulting problem, the TR-Benders framework is extended to leverage the additional decomposition structure introduced by the multiperiod formulation, enabling parallel solution across both planning periods and representative operating days. The approach further incorporates per-period trust-region constraints and a relaxation-based warm-start to improve convergence. The methodology is demonstrated on a high-resolution 2000-bus Texas grid model, producing an investment plan for 2030-2045 with hourly operational decisions and a final optimality gap $<1.0\%$.

\section{Formulation}
\label{sec:formulation}
The all-years multiperiod (AY) TEP+Storage model is formulated as a two-stage optimization program. The first stage determines the timing and location of investment decisions, including transmission upgrades and storage installations, to minimize capital expenditures $c^{\mathrm{cap}}(\cdot)$. The second stage resolves hourly operational decisions to minimize operational costs $c^{\mathrm{op}}(\cdot)$. The formulation extends the single-period planning framework \cite{wu2025high} to a multiperiod setting in order to capture the delayed value of investments, and is summarized in Model~\ref{model:TNEP_ptdf}.


\subsection{Investment Decisions and Variables}
Transmission upgrades are modeled as reconductoring of pre-existing transmission lines, which has been shown to be frequently favored over new line construction due to its cost-effectiveness \cite{chojkiewicz2024accelerating_reconductoring_pnas}. The cumulative upgrade of branch $ij$ in period $y\in \Y$ is denoted by binary variable $\gamma_{ij, y}$, which equals $1$ if the branch is reconductored, and $0$ otherwise. Each upgrade corresponds to a $\Delta^c_{ij}$ MW increase in capacity. 

Energy storage is modeled as short-duration lithium-ion battery technology, given its technological maturity and widespread adoption in power systems \cite{NRELBatteryCost}. Cumulative storage investment at node $i$ in period $y$ is denoted by decision variable $\sigma_{i, y}$, representing the installed energy capacity (MWh). Batteries are modeled as four-hour duration, with maximum charge/discharge power of $\sigma_i/4$. 

For notational convenience, let $\gamma_y$ denote the vector of transmission upgrade decisions across all branches in period $y$, and $\sigma_y$ denote the vector of storage investment decisions across all nodes in period $y$. Let $\gamma := (\gamma_y)_{y\in\Y}$ and $\sigma := (\sigma_y)_{y\in\Y}$ denote the concatenation of these decisions across all planning periods. The full vector of investment decisions is denoted by $x := (\gamma, \sigma)$, and $x_y := (\gamma_y, \sigma_y)$ denotes the corresponding investment decisions in period $y$.

\subsection{Objective}
The objective of the AY TEP+Storage problem is to minimize the sum of capital and operational expenditures across all planning periods. The delayed value of investments is modeled using an annual discount factor $d_y := (1+\alpha)^{-p(y)}$, where $\alpha$ is the annual interest rate and $p(y)$ denotes the number of years elapsed since the reference year for the start of planning period $y$, with $p(1)=0$. For example, with 5-year planning periods, $p(2)=5$, $p(3)=10$, and so on. Capital investment costs for line reconductoring, $c_{ij}^\text{rec}$, and storage, $c_i^\text{stor}$, are expressed as annualized costs. The resulting capital expenditure is a linear function of the cumulative line reconductoring decisions $\gamma$ and storage investment decisions $\sigma$ made by each planning period $y$. Both investment variables are initialized to zero in the baseline period $y=0$. 
\begin{align}
\label{eq:TNEP:capex}
c^{\mathrm{cap}}(\gamma,\sigma)
&:= \sum_{y \in \Y} d_y \Bigg[
      \sum_{ij \in \E} c_{ij}^{\mathrm{rec}}
      \bigl(\gamma_{ij,y}-\gamma_{ij,y-1}\bigr)
      \notag \\
&\qquad\qquad
    + \sum_{i \in \N} c_i^{\mathrm{stor}}
      \bigl(\sigma_{i,y}-\sigma_{i,y-1}\bigr)
\Bigg].
\end{align}
Given planning decisions $x_y := (\gamma_y, \sigma_y)$ for period $y$, the operational expenditure incurred in planning period $y$ on representative day $r\in \R$ is denoted by $c^{op}_{y,r}(\cdot)$. The operating cost includes generation costs, storage charging and discharging costs, and a penalty on load-shed slack variables $\xi$.

\begin{align*}
c^{\mathrm{op}}_{y,r}(x_y)
&:= \sum_{i \in \N} \sum_{t \in \T} \Big[
      c_i^{\mathrm{G}}\!\left(p^g_{i,y,r,t}\right)
      \notag \\
&\qquad\qquad\quad
      + c_i^{\mathrm{S}}\!\left(\mathrm{ch}_{i,y,r,t}
      + \mathrm{dis}_{i,y,r,t}\right)
      + \lambda\, \xi_{i,y,r,t}
    \Big].
\end{align*}
The operational costs are also annualized to maintain consistent units with the annualized investment costs. Specifically, the operational cost for planning period $y$ is computed as the probability-weighted sum over representative days and scaled by 365 days per year:
\begin{align*}
    c^{\mathrm{op}}_{y}(x_y) := \textstyle{ 365 \sum_{r \in \R}} \;\Omega_{\R}(r) \; c^{\mathrm{op}}_{y,r}(x_y)
\end{align*}
The aggregate operational cost in the AY formulation is then obtained by summing across all planning periods, given the full vector of investments across all periods $x$.
\begin{align}
    \label{eq:TNEP:opex_all}
    c^{\mathrm{op}}(x) := \textstyle{\sum_{y \in \Y}} \; d_y \; c^{\mathrm{op}}_{y}(x_y)
\end{align}

\subsection{Constraints}
A DC power flow is used, via the Power Transfer Distribution Factor (PTDF) matrix $\Phi$, which maps net nodal injections to transmission line flows $\pf$. For each planning period $y$, representative day $r$, and hour $t$, the vector of net nodal injections is
\begin{align*}
\text{inj}^{y,r,t}
&:= \bigl[
      \pg_{i,y,r,t}
      + \text{dis}_{i,y,r,t}
      \notag \\
&\qquad
      - \pd_{i,y,r,t}
      - \text{ch}_{i,y,r,t}
      + \xi_{i,y,r,t}
    \bigr]^{N}_{i=1}
\end{align*}
The operational constraints are as follows:
\begin{align}
    \label{eq:TNEP:ptdf_flow}
    [\pf_{ij,y,r,t}]^{E}_{ij=1} &= \Phi \times \text{inj}^{y,r,t} \\
    \label{eq:TNEP:global_power_balance}
    \mathbf{1}^\top \text{inj}^{y,r,t} &= 0 \\
    \label{eq:TNEP:thermal_limit}
    -\pfmax_{ij} - \gamma_{ij,y} \Delta^c_{ij} & \leq \pf_{ij,y,r,t} \leq \pfmax_{ij} +\gamma_{ij,y} \Delta^c_{ij} \\
    \label{eq:TNEP:generation:min_max_limits}
    \pgmin_{i,y,r,t} & \leq \pg_{i,y,r,t} \leq \pgmax_{i,y,r,t} \\
    \label{eq:TNEP:state_of_charge}
    \text{soc}_{i,y,r,t} &= \text{soc}_{i,y,r,t-1} + \text{ch}_{i,y,r,t}\eta - \text{dis}_{i,y,r,t}/\eta \\
    \label{eq:TNEP:initial_charge}
    \text{soc}_{i,y,r,1} &= \text{ch}_{i,y,r,1}\eta - \text{dis}_{i,y,r,1}/\eta \\
    \label{eq:TNEP:final_charge}
    \text{soc}_{i,y,r,T} &= 0 \\
    \label{eq:TNEP:soc_energy_rating}
    0 & \leq \text{soc}_{i,y,r,t} \leq \sigma_{i,y} \\
    \label{eq:TNEP:charge_limits}
    0 & \leq \text{ch}_{i,y,r,t} \leq \sigma_{i,y}/4 \\ 
    \label{eq:TNEP:discharge_limits}
    0 & \leq \text{dis}_{i,y,r,t} \leq \sigma_{i,y}/4
\end{align}

Constraint \eqref{eq:TNEP:ptdf_flow} determines power flows from nodal injections, while \eqref{eq:TNEP:global_power_balance} enforces system-wide power balance. Thermal limits on line flows are imposed by \eqref{eq:TNEP:thermal_limit}. Generator dispatch is constrained by \eqref{eq:TNEP:generation:min_max_limits}, where renewable units are modeled with zero minimum output and scenario-dependent upper bounds to allow curtailment, while conventional generators are dispatchable within their operating limits. Storage operation is captured by \eqref{eq:TNEP:state_of_charge}–\eqref{eq:TNEP:discharge_limits}, which describe SoC dynamics with efficiency losses, zero SoC conditions at the beginning and end of each representative day, and charging and discharging limits. Consequently, storage operation is constrained within individual representative days, preserving inter-day independence but limiting the model's ability to capture inter-day energy shifting. Finally, to reflect the irreversibility of transmission and storage investments, monotonicity constraints are enforced:
Constraint \eqref{eq:TNEP:ptdf_flow} determines power flows from nodal injections, while \eqref{eq:TNEP:global_power_balance} enforces system-wide power balance. Thermal limits on line flows are imposed by \eqref{eq:TNEP:thermal_limit}. Generator dispatch is constrained by \eqref{eq:TNEP:generation:min_max_limits}, where renewable units are modeled with zero minimum output and scenario-dependent upper bounds to allow curtailment, while conventional generators are dispatchable within their operating limits. Storage operation is captured by \eqref{eq:TNEP:state_of_charge}–\eqref{eq:TNEP:discharge_limits}, which describe SoC dynamics with efficiency losses, zero SoC conditions at the beginning and end of each representative day, and charging and discharging limits. Consequently, storage operation is constrained within individual representative days, preserving inter-day independence but limiting the model's ability to capture inter-day energy shifting. Finally, to reflect the irreversibility of transmission and storage investments, monotonicity constraints are enforced:
\begin{align}
    \label{eq:TNEP:monotonicity_lines_storage}
    x_{y-1} & \preceq x_{y}
\end{align}

\begin{model}[!t]
    \caption{Multiperiod TEP+Storage Formulation}
    \label{model:TNEP_ptdf}
    \begin{align*}
        \text{min} \quad 
            & \eqref{eq:TNEP:capex} + \eqref{eq:TNEP:opex_all}
            \\
        \text{s.t.} \quad 
        & \eqref{eq:TNEP:ptdf_flow}-\eqref{eq:TNEP:global_power_balance} , \quad \forall y \in \Y, r \in \R, t \in \T \\
        & \eqref{eq:TNEP:thermal_limit},
            \quad \forall ij \in \E, y \in \Y, r \in \R, t \in \T \\ 
        & \eqref{eq:TNEP:generation:min_max_limits},
            \quad \forall i \in \N, y \in \Y, r \in \R, t \in \T \\
        & \eqref{eq:TNEP:state_of_charge}
            \quad \forall i \in \N, y \in \Y, r \in \R, t \in \{2,...,T\} \\
        & \eqref{eq:TNEP:initial_charge}-\eqref{eq:TNEP:final_charge},
            \quad \forall i \in \N, y \in \Y, r \in \R \\
        & \eqref{eq:TNEP:soc_energy_rating}-\eqref{eq:TNEP:discharge_limits},
        \quad \forall i \in \N, y \in \Y, r \in \R, t \in \T \\
        & \eqref{eq:TNEP:monotonicity_lines_storage}, \quad y \in \Y
    \end{align*}
\end{model}

\noindent \emph{Remark:} Reconductoring increases thermal limits while leaving line electrical parameters largely unchanged \cite{Henriksen2013ConductorPerformance}, permitting the use of the PTDF-based power flow model which assumes fixed network topology and electrical parameters.

The size of Model \ref{model:TNEP_ptdf} grows linearly with the number of planning periods and representative days ($Y \times R$), necessitating scalable solution strategies for large instances.

\section{Methodology}
\label{sec:methodology}
To solve the multiperiod formulation presented in Model \ref{model:TNEP_ptdf}, a multicut Benders decomposition (BD) framework is employed. In a baseline multicut BD, Model \ref{model:TNEP_ptdf} is decomposed into a master problem (MP) that determines first-stage investment decisions and a collection of subproblems (SPs) that evaluate operational decisions under fixed investments. At each iteration, dual information from each SP is used to generate a corresponding cut, which is added to the MP until convergence. However, baseline implementations struggle to converge, with significant oscillations in MP iterates and dual degeneracy across investment configurations. 

To address these computational challenges, this work develops a trust-region (TR) stabilization and relaxation-based warm-start strategy for BD applied to grid planning problems. While these techniques were previously applied in either single-period expansion models \cite{wu2025high} or simplified multiperiod settings \cite{goke2024stabilized}, the present work extends them to a multiperiod all-years formulation by introducing \textbf{per-period trust-region constraints} on investment decisions. The multiperiod all-years (AY) TR-Benders procedure is described next, and presented in detail in Appendix~\ref{app:algorithm} (Algorithm \ref{alg:tr_benders}).


\subsection{AY TR-Benders Master Problem (MP)}
\label{subsec:master}
The MP determines transmission reconductoring and storage investment decisions for each planning period. Its objective includes auxiliary variables $\theta_{y,r}$ that provide lower bounds on the operational cost $c^{\text{op}}_{y,r}(\cdot)$ for each representative day and period. In addition, a proxy variable $\rho$ underestimates system-wide load shedding across all periods and is constrained to zero via \eqref{eq:TNEP:no_load_shed}.
At each iteration, $
\theta$ and $\rho$ are refined through Benders cuts and accumulated via \eqref{eq:TNEP:accumulated_cuts}. The MP at iteration $k$ is given in Model \ref{model:TNEP_ptdf_master}. \\

\noindent \emph{Remark:} Both Models \ref{model:TNEP_ptdf} and \ref{model:TNEP_ptdf_master} include load-shed slack variables $\xi$ to ensure recourse feasibility. The two-stage formulation relies on a penalty term to discourage load shedding, whereas the Benders formulation instead uses the master proxy variable $\rho$ to enforce zero load shedding in the final solution, avoiding the need to tune the penalty parameter as the number of representative days increases.

\begin{model}[!t]
    \caption{Benders Master Problem (MP)}
    \label{model:TNEP_ptdf_master}
    \begin{align}
        z^{\text{MP}}_k = \text{min} \quad 
            & \eqref{eq:TNEP:capex} + \textstyle\sum_{y \in \Y, r\in \R}\; d_y \Omega_{\R}(r)\theta_{y,r} + \rho \nonumber
            \\
        \text{s.t.} \quad 
        & \eqref{eq:TNEP:monotonicity_lines_storage}, \quad y \in \Y \\
        \label{eq:TNEP:no_load_shed}
        & \rho = 0 \\
        \label{eq:TNEP:accumulated_cuts}
        & \textstyle\bigcup_{y \in \Y, r\in \R}\mathcal{C}_{y,r}^k 
    \end{align}
\end{model}

\begin{model}[!t]
    \caption{Benders Subproblem (SP)}
    \label{model:TNEP_ptdf_subproblem}
    \begin{align}
        \text{min} \quad 
            & c^{\text{op}}_{y,r}(\gamma^k_y, \sigma^k_y) \nonumber \\
        \text{s.t.} \quad 
        & \eqref{eq:TNEP:ptdf_flow}-\eqref{eq:TNEP:global_power_balance} , \quad \forall t \in \T \nonumber \\
        & \eqref{eq:TNEP:thermal_limit},
            \quad \forall ij \in \E, t \in \T \nonumber \\ 
        & \eqref{eq:TNEP:generation:min_max_limits},
            \quad \forall i \in \N, t \in \T \nonumber \\
        & \eqref{eq:TNEP:state_of_charge},
            \quad \forall i \in \N, t \in \{2,...,T\} \nonumber \\
        & \eqref{eq:TNEP:initial_charge}-\eqref{eq:TNEP:final_charge},
            \quad \forall i \in \N \nonumber \\
        & \eqref{eq:TNEP:soc_energy_rating} - \eqref{eq:TNEP:discharge_limits},
        \quad \forall i \in \N, t \in \T \nonumber \\
        & \gamma_{ij} = \gamma^k_{ij} \quad (\pi^{y,r,k}_\gamma), \quad \forall ij \in \E  \label{eq:TNEP:fixed_line} \\
        & \sigma_i = \sigma^k_i \quad (\pi^{y,r,k}_\sigma), \quad \forall i \in \N  \label{eq:TNEP:fixed_stor}
    \end{align}
\end{model}

\subsection{Benders Subproblem (SP)}
After the $k$th iteration of an MP solve, the investment configuration solution $x^k$ is broadcast to all the SPs. Constraints \eqref{eq:TNEP:fixed_line}-\eqref{eq:TNEP:fixed_stor} enforce the investments to match the MP solution at period $y$. Their dual variables, $\pi^{y,r,k}_\gamma$ and $\pi^{y,r,k}_\sigma$, respectively, are used to construct the Benders cuts. 

Each SP is solved using a two-phase procedure. In Phase I, load feasibility is enforced by minimizing total load shedding, $\sum_{i \in \N, t \in \T} \xi_{i,y,r,t}$. If nonzero load shedding is detected, the SP terminates and generates a \emph{feasibility cut} for the load-shed underestimator $\rho$. Otherwise, Phase II minimizes the operational dispatch cost (Model \ref{model:TNEP_ptdf_subproblem}), yielding an \emph{optimality cut} for the operational cost underestimator $\theta_{y,r}$. The feasibility and optimality cuts are, respectively:
\begin{align}
    \label{eq:TNEP:feasibility_cut}
    \rho & \geq \textstyle{\sum_{i \in \N, t\in \T}}\;\xi_{i,y,r,t} + (\pi^{y,r,k})^\top(x_{y,r} - x_{y,r}^k) \\
    \label{eq:TNEP:optimality_cut}
    \theta_{y,r} & \geq c^{\text{op}}_{y,r}(x^k_{y,r}) + (\pi^{y,r,k})^\top(x_{y,r} - x_{y,r}^k)
\end{align}

To further accelerate SP solves, constraints \eqref{eq:TNEP:ptdf_flow} and \eqref{eq:TNEP:thermal_limit} are imposed via lazy constraint generation under the PTDF formulation, yielding an equivalent but more tractable SP. Details of implementation are shown in Appendix \ref{app:lazy_constraint_generation}.

\subsection{Per-period Stabilization}
Baseline Benders implementations are susceptible to oscillatory MP iterates and dual degeneracy when multiple transmission–storage configurations yield similar operating costs. To mitigate these challenges in the multiperiod setting, the MP is stabilized using $l_1$ trust region constraints \textbf{applied in a per-period basis}. The investments are stabilized about an anchor point $(\breve\gamma, \breve\sigma)$ that corresponds to the incumbent investment trajectory achieving the best known upper bound. The trust region constraints are:
\begin{align}
    \label{eq:TNEP:line_trust_region}
    \textstyle\sum_{ij\in\mathcal{E}} \big|\gamma_{ij, y}-\breve\gamma_{ij, y}\big| & \leq \mathrm{rad}_\gamma, \quad \forall y \in \Y \\
    \label{eq:TNEP:storage_trust_region}
    \textstyle\sum_{i\in\mathcal{N}} \big|\sigma_{i, y}-\breve\sigma_{i, y}\big| & \leq \mathrm{rad}_\sigma, \quad \forall y \in \Y,
\end{align}
where $\rad_\gamma$ and $\rad_\sigma$ denote the transmission and storage TR radii, respectively. The specific parameter values are discussed in Section \ref{subsec:hyperparameters}.

In addition to improving convergence, the TR stabilization accelerates SP solves by restricting the search to nearby investment configurations, leading to more consistent congestion patterns across iterations. This, in turn, enhances the effectiveness of the lazy constraint generation by reducing the number of newly violated flow constraints.

The MP is further strengthened using provably valid inequalities that impose per-period minimum required storage capacity to ensure load feasibility, thereby improving convergence performance; their derivation and proofs are provided in Appendix \ref{app:valid_storage_cap_lb}.

\subsection{Benders Warm-start (WS)}
\label{subsec:benders_warmstart}
To initialize the MP away from trivial zero-investment solutions with extreme load shedding, a relaxation-based warm-start (WS) procedure is employed. The WS is constructed by solving continuous relaxations of Model \ref{model:TNEP_ptdf} sequentially across planning periods, enforcing monotonic investment decisions between periods. For each period $y\in\Y$, the continuous relaxation is solved independently for all representative days $r\in\R$, yielding relaxed investments $(\gamma^{y,r}, \sigma^{y,r})$. These are aggregated by taking the component-wise maximum across representative days and rounding up to the admissible investment granularity:
\begin{align*}
    \bar\gamma_{y} := \textstyle\big\lceil \max_{r\in\R} \gamma^{y,r} \big\rceil, \quad
    \bar\sigma_{y} := \textstyle\big\lceil \max_{r\in\R} \sigma^{y,r} \big\rceil .
\end{align*}
The aggregated investments are enforced in subsequent periods via monotonicity constraints,
\begin{align*}
\gamma^{y,r} \ge \bar\gamma_{y-1}, \qquad \sigma^{y,r} \ge \bar\sigma_{y-1}, \quad \forall r\in\R.
\end{align*}

This procedure yields a warm-started investment trajectory consistent with intertemporal monotonicity \eqref{eq:TNEP:monotonicity_lines_storage} and worst-case operational requirements across representative days. The BD is initialized at $(\bar\gamma,\bar\sigma)$.

To further reduce degeneracy, investment decisions are restricted to candidate sets inferred from the relaxed solutions in the final planning period:
\begin{align*}
    \mathcal{E}_c &:= \{\, ij\in\mathcal{E} \mid \exists r\in\R:\ \gamma^{Y,r}_{ij,Y} > 0 \}, \\
    \mathcal{N}_c &:= \{\, i\in\mathcal{N} \mid \exists r\in\R:\ \sigma^{Y,r}_{i,Y} > 0 \}.
\end{align*}

\subsection{Convergence of TR-Benders}
In baseline BD, the MP optimum $z_k^\text{MP}$ provides a valid global lower bound ($\text{LB}$), representing a theoretical minimum system cost over all feasible investment configurations. With the introduction of TR constraints, the MP is restricted to a local neighborhood, rendering $z_k^\text{MP}$ only a \textit{local} LB. A valid global lower bound $\text{LB}^*$ is therefore constructed separately; its construction is presented in Appendix \ref{subsec:global_lower_bound_lb}. 

Whenever all SPs are solved and no load shedding is observed (i.e. $\sum_{i \in \N, t\in \T} \xi_{i,y,r,t} = 0 \; \forall y\in\Y, r\in\R$), the resulting solution is feasible and yields objective $\hat z^k := c^{\mathrm{cap}}(x^k) + c^{\text{op}}(x^k)$. The upper bound (UB) is then updated as $\mathrm{UB} := \min\{\mathrm{UB}, \hat z_k\}$. 

During TR-Benders, the current UB is also used to enforce a strict level-set constraint that ensures the MP explores strictly better investment configurations.
\begin{align}
    \label{eq:TNEP:level_set}
    z_k^\MP < \UB
\end{align}
In practice, the level-set constraint is relaxed to $z_k^\MP \leq \UB - \delta$, where $\delta > 0$ is a small numerical tolerance.

The algorithm terminates once the relative optimality gap satisfies $(\mathrm{UB} - \mathrm{LB}^*)/\mathrm{UB} \leq \epsilon$ for a specified tolerance $\epsilon > 0$. \\

\section{Case Study}
\label{sec:casestudy}
The proposed methodology is evaluated on the ACTIVSg2000 test system, a large-scale 2000-bus representation of the Texas grid developed using publicly available data and statistical representations of real-world power systems \cite{synthetic_1}. Investment planning periods are considered every 5 years, $\Y = \{2030, 2035, 2040, 2045\}$, with an annual interest rate of $\alpha = 0.05$. For each planning period, system operations are evaluated at the end of the period after investment decisions have been implemented; for example, operations for period $y=1$ (2030–2034) are evaluated in year 2030.

\subsection{Network Data}
The ACTIVSg2000 test system has 2000 buses, 3206 transmission lines, 423 conventional generators, and 175 renewable generators. A PTDF cutoff of 0.005 is applied to sparsify the network matrix, a common technique used to remove negligible sensitivities and improve computational tractability \cite{rahmani2016comprehensive}. Generation capacity projections are based on the EIA Annual Energy Outlook \cite{EIA2023Outlook}, with all resources scaled uniformly to account for expected capacity additions and retirements. A 1.5\% annual load growth is applied uniformly across all buses. Detailed generation, load, and discounting projections for the case study are provided in Appendix~\ref{app:case_study_parameters}. 

A set of 18 representative days ($R = 18$), each spanning 24 hours, is selected from a full year of load, wind, and solar data. The representative days (listed in Appendix \ref{app:representative_days}) are chosen using a max–min diversity criterion applied to normalized daily profiles, to identify a diverse set of operating conditions, and are assigned equal probabilities.

\subsection{Investment Parameters}
Line upgrades via reconductoring are assumed to double thermal limits \cite{chojkiewicz2024accelerating_reconductoring_pnas}, giving $\Delta^c_{ij} = \pfmax_{ij}$. Storage investments are restricted to discrete units with an energy rating of 250 MWh each, with up to 12 units permitted per node (3 GWh total). A charge and discharge efficiency of $\eta = 95\%$ is assumed, yielding a round-trip efficiency of $90.25\%$. Investment cost assumptions for the case study are summarized in Appendix~\ref{app:case_study_parameters}.
Reconductoring costs are adapted from \cite{chojkiewicz2024accelerating_reconductoring_pnas}, and storage costs are adapted from \cite{NRELBatteryCost}. Annualized capital expenditures are computed by applying a capital recovery factor (CRF) to the overnight investment cost, based on the interest rate $\alpha$ and asset lifetime $l$, as $\text{CRF} := \frac{\alpha(1+\alpha)^l}{(1 + \alpha)^l - 1} $. Battery storage is assumed to have a 15-year lifetime, consistent with the median value reported in published battery storage cost projections \cite{NRELBatteryCost}. Although reconductored transmission lines typically exhibit lifetimes of 40–60 years \cite{accc_longevity}, the case study is restricted to the 2030–2050 planning horizon; therefore, an effective lifetime of 20 years is assumed for line upgrades.

\subsection{Comparison with Rolling Horizon (RH)}
One common approach to multiperiod planning is the use of a rolling horizon (RH) framework, in which single-period models are solved sequentially. To compare against the AY formulation, the case study is also solved using RH. Under RH, investment decisions are determined period by period and are myopic, as future-period projections are not anticipated. For each planning period, a separate instance of the TR-Benders algorithm is solved. The WS procedure in Sec.~\ref{subsec:benders_warmstart} is modified such that the investments enforced in period $y$ are taken from the optimal single-period TR-Benders solution in period $y-1$, causing RH to alternate between WS and TR-Benders across periods. In contrast, the AY methodology performs a single WS phase to initialize a joint multiperiod investment trajectory, followed by a single execution of TR-Benders that jointly optimizes investments across all periods.

\subsection{Computational Setup and Hyperparameters}
\label{subsec:hyperparameters}
All computational experiments are conducted using a Julia implementation based on JuMP v1.23.2 \cite{Lubin2023_jump}, with optimization performed using Gurobi v1.3.1 \cite{gurobi} on the Georgia Tech PACE Phoenix cluster \cite{PACE} (24-core Intel Xeon nodes, Linux). Each worker is allocated 12GB of RAM, consistent with the 288GB memory available per node. 

The TR hyperparameters were selected based on preliminary computational experiments. The storage TR radius was fixed at $\mathrm{rad}_\sigma=1$ due to the greater degeneracy observed in storage investment decisions, while the transmission TR radius was initialized at $\mathrm{rad}_\gamma=1$ and updated dynamically (Appendix \ref{app:algorithm}) to balance stabilization and exploration of promising transmission expansion plans. For all WS procedures, the load-shed penalty was set to $\lambda = 50{,}000\$/\text{MWh}$, which yielded zero load-shed solutions across all tested instances.

The RH experiments are run with sequential time limits of 24/24/24/72 hours, while the AY formulation is solved with a time limit of 144 hours.

\section{Results}
\label{sec:results}

\subsection{All-Years Multiperiod (AY) Computational Performance}
The AY warm-start procedure, denoted ``AY-WS'' required 4.2 hours of runtime and produced an objective value of \$29.14B with an optimality gap of 3.63\%. The AY TR-Benders algorithm was then executed with a 144-hour time limit, resulting in a total runtime of 148.2 hours. The resulting investment plan, denoted ``AY-Final'' achieved a total cost of \$28.25B with a final optimality gap of 0.59\%. Figure \ref{fig:gap_over_time} plots the optimality gap across total runtime of the AY method, which includes WS runtime. AY attains a solution quality of under 1\% optimality gap within 64 hours, illustrating the method's ability to solve high-resolution planning instances.

Figure~\ref{fig:master_subproblem_solve_times} compares the MP and SP solve times over the 338 AY TR-Benders iterations for the representative day of August 11, 2045. The reported SP solve times include both the Phase I (load-shed minimization) and, when applicable, the subsequent Phase II (operational cost minimization) solves. This day corresponds to the highest system load in the year and occurs in the final planning period (2045), which also has the highest penetration of renewable generation. These operating conditions induce severe congestion across the network and require substantial storage utilization to satisfy demand, making this day the most computationally challenging among all representative days and planning periods. Consequently, it produces the longest SP solve times and determines the duration of each TR-Benders iteration, since progress cannot occur until all SPs are evaluated.

On average, the MP solves in 0.43 minutes, whereas the SPs require an average of 25.1 minutes. The MP solves account for only 1.7\% of the TR-Benders runtime. This highlights the relatively low computational complexity of the MP, typical to grid planning problems \cite{goke2024stabilized,wu2025high}, and is further mitigated by trust-region stabilization. Although extending single-period BD to the multiperiod setting scales the MP size by $Y$, the TR constraints restrict the search to a localized region of the solution space. The resulting MP remains computationally lightweight even under a multicut implementation, where each iteration introduces up to $Y\times R$ cuts. These results demonstrate that the proposed AY methodology scales well with respect to the MP: increasing the number of representative days $R$ is largely inconsequential due to parallelizable SP solves, while scaling the number of planning periods $Y$ remains tractable given the consistently fast MP solve times. Moreover, this efficiency enables the MP to accommodate additional modeling features, such as budget constraints, decarbonization targets, and emissions policies.

\subsection{Comparison with Rolling Horizon (RH) Formulation}
\label{subsec:comparison_rh_rolling_horizon}
Table \ref{tab:obj_optimality_gaps} compares the computational performance and achieved optimality gaps of the AY methodology with the RH approach. ``RH-Final'' refers to the final investment plan produced by RH after all periods have been executed sequentially. Reported runtimes represent total wall-clock time expended, including warmstart. Final solves were terminated upon reaching prescribed time limits; thus, reported runtimes reflect imposed limits rather than time to proven optimality.

\begin{table}[t]
    \caption{Computational Performance and Optimality Gaps}
    \label{tab:obj_optimality_gaps}
    \centering
    \renewcommand{\arraystretch}{1.2}
    \begin{tabular}{|c|c|c|c|}
    \hline
    \textbf{Method} 
        & \textbf{\parbox[t]{1.4cm}{\centering Runtime (h)\\[0.4ex]}} 
        & \textbf{\parbox[t]{1.2cm}{\centering Best UB\\(\$B)}} 
        & \textbf{\parbox[t]{1.2cm}{\centering Gap\\(\%)}}  \\
    \hline
    RH-Final
        & 147.7
        & 28.31
        & 0.81 \\
    AY-WS 
        & 4.2
        & 29.14
        & 3.63 \\
    AY-Final
        & 148.2
        & \textbf{28.25}
        & \textbf{0.59} \\
    \hline
    \end{tabular}
\end{table}

RH-Final attains a high-quality solution with an optimality gap of 0.81\% in 147.7 hours. In comparison, AY achieves comparable solution quality in less than half the runtime, reaching an optimality gap below 0.81\% within 72 hours (Fig.~\ref{fig:gap_over_time}). This reduction in solution time reflects the increased parallelism of the AY formulation compared to RH.

The optimality gap obtained by RH-Final and AY-Final are, however, comparable. This outcome is likely influenced by the structure of the case study: although generation types and load are geographically heterogeneous, future-year load and generation projections are applied uniformly across the network. Consequently, the case study does not capture spatially and temporally heterogeneous shifts (e.g. a large computing load emerging at one node in 2030 followed by the introduction of a new generator at another node in 2035) which may mitigate the otherwise myopic limitations of the RH approach by reinforcing consistent intertemporal trends.

\begin{figure}
    \centering
    \includegraphics[width=0.9\linewidth, trim={4 10 4 7}, clip]{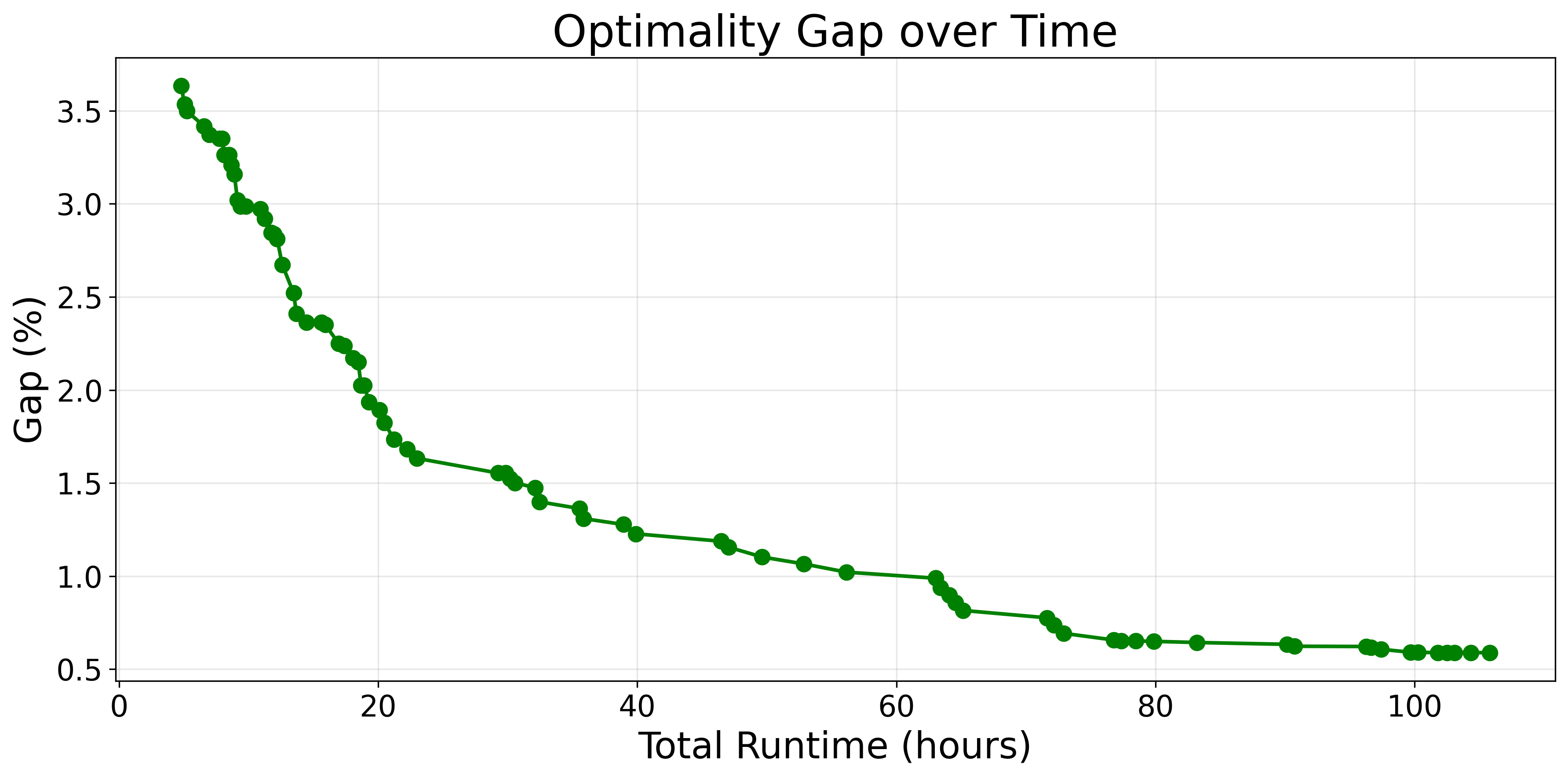}
    \caption{AY TR-Benders optimality gap over total runtime (including WS)}
    \label{fig:gap_over_time}
\end{figure}

\begin{figure}
    \centering
    \includegraphics[width=0.75\linewidth, trim={4 10 4 7}, clip]{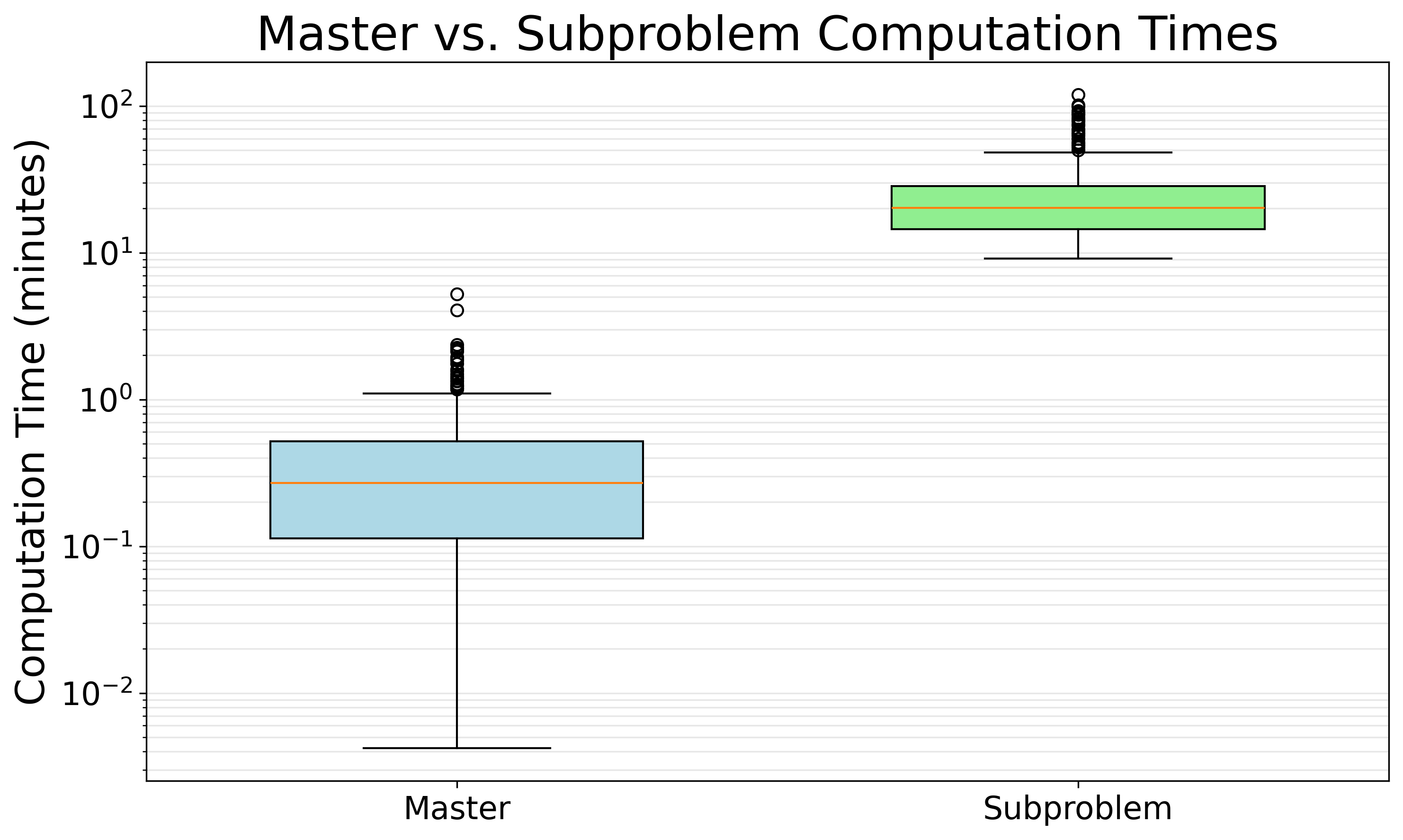}
    \caption{AY TR-Benders MP and SP computation times}
    \label{fig:master_subproblem_solve_times}
\end{figure}

\subsection{Investment Result Comparison}
\label{subsec:rolling_horizon_results}
Table~\ref{tab:first_stage_investments} compares cumulative first-stage investment decisions under RH-Final and AY-Final, while Table~\ref{tab:first_stage_investments_costs} reports the corresponding solution costs. Incremental first-stage investments are reported in undiscounted terms to reflect physical build decisions, whereas total period costs are discounted and include operational costs.

Relative to the RH, the AY invests more cumulatively in transmission reconductoring and less in storage capacity. This shift results in a lower final objective value and optimality gap, reflecting the higher unit cost of storage relative to reconductoring (Appendix \ref{app:case_study_parameters}). Moreover, incremental investment patterns show that AY front-loads reconductoring investments more aggressively than RH, reducing the need for storage investments in later periods. This behavior is consistent with the multiperiod perspective of the AY, under which earlier transmission upgrades can be economically justified in anticipation of increasing future system requirements.

Notably, both RH-Final and AY-Final incur almost identical total period costs in 2030 despite differing line configurations, highlighting the degeneracy of investment configurations with similar costs. Overall, these results indicate that the multiperiod formulation better internalizes future system needs and consequently favors earlier transmission investments.

\begin{table}[t]
    \caption{Cumulative First-Stage Investments}
    \label{tab:first_stage_investments}
    \centering
    \renewcommand{\arraystretch}{1.2}
    \begin{tabular}{|c|c|c|c|c|c|c|}
    \hline
    \textbf{Method}
    & \textbf{Quantity} 
    & \textbf{2030} 
    & \textbf{2035} 
    & \textbf{2040} 
    & \textbf{2045} \\
    \hline
    \textbf{RH-Final} & \# Lines
        & $214$ 
        & $333$ 
        & $453$ 
        & $527$ \\
    & \# Storage Units 
        & $0$ 
        & $26$ 
        & $83$ 
        & $148$ \\
    & \parbox[t]{2.0cm}{\centering Storage Capacity\\(GWh)\\[0.4ex]} 
        & $0$ 
        & $28.00$ 
        & $96.75$ 
        & $167.50$ \\
    \hline
    \textbf{AY-Final} & \# Lines
        & $299$ 
        & $398$ 
        & $502$ 
        & $568$ \\
        & \# Storage Units 
        & $0$ 
        & $32$ 
        & $77$ 
        & $138$ \\
    & \parbox[t]{2.0cm}{\centering Storage Capacity\\(GWh)\\[0.4ex]} 
        & $0$ 
        & $27.75$ 
        & $96.50$ 
        & $165.75$ \\
    \hline
    \end{tabular}
\end{table}

\begin{table}[t]
    \caption{Incremental First-Stage Investment and Total Period Cost (\$)}
    \label{tab:first_stage_investments_costs}
    \centering
    \renewcommand{\arraystretch}{1.2}
    \begin{tabular}{|c|c|c|c|c|c|c|}
    \hline
    \textbf{Method}
    & \textbf{Costs} 
    & \textbf{2030} 
    & \textbf{2035} 
    & \textbf{2040} 
    & \textbf{2045} \\
    \hline
    \textbf{RH-Final} & Incr. Line
        & $67.3$M 
        & $59.8$M 
        & $27.4$M 
        & $31.5$M \\
    & Incr. Storage 
        & $0$ 
        & $2.70$B 
        & $6.62$B 
        & $6.81$B \\
    & \parbox[t]{2.0cm}{\centering Total Period \\(discounted)\\[0.4ex]} 
        & $5.83$B 
        & $6.97$B 
        & $8.36$B 
        & $7.15$B \\
    \hline
    \textbf{AY-Final} & Incr. Line
        & $91.9$M 
        & $71.4$M 
        & $65.3$M 
        & $25.1$M \\
        & Incr. Storage
        & $0$ 
        & $2.67$B 
        & $6.62$B 
        & $6.67$B \\
    & \parbox[t]{2.0cm}{\centering Total Period\\(discounted)\\[0.4ex]} 
        & $5.84$B 
        & $6.96$B 
        & $8.37$B 
        & $7.07$B \\
    \hline
    \end{tabular}
\end{table}

\section{Conclusion}
\label{sec:conclusion}

This work proposes an all-years (AY) multiperiod formulation for the TEP+Storage problem. The formulation is solved using a trust-region–stabilized Benders decomposition. New per-period trust-region constraints are introduced, complemented by a relaxation-based warm-start procedure. The proposed method is evaluated on a 2,000-bus synthetic Texas system and achieves a final optimality gap of 0.59\%. Compared with a rolling-horizon (RH) approach, the AY approach attains comparable or improved solution quality with substantially reduced computational effort, benefiting from joint optimization across periods and increased parallelism. The consistently low computational burden of the master problem further indicates favorable scalability to larger planning horizons and model extensions.

The AY formulation also favors earlier transmission reconductoring, leading to lower cumulative investment in costly storage relative to RH and illustrating the advantages of joint multiperiod optimization. Future work will investigate enhanced stabilization strategies, such as one-sided trust-region constraints informed by conservative warm-starts, as well as more heterogeneous future projection scenarios, temporal linking of storage state-of-charge across representative days, and additional parallel procedures for generating valid inequalities.

\section{AI Usage Disclosure}
AI tools were used to assist with code implementation and debugging, as well as for editorial support to improve clarity and writing style. AI tools were not used to generate research results, derive models, produce figures, or create original scientific content.

\appendix 
\subsection{Multiperiod All-Years TR-Benders Algorithm} \label{app:algorithm}

Algorithm \ref{alg:tr_benders} describes the implementation of the proposed AY TR-Benders algorithm.

\begin{algorithm}[t]
\caption{AY TR-Benders}
\label{alg:tr_benders}
\begin{algorithmic}[1]
\Require $\Y,\R$; $\epsilon$; $\LB^*$; $(\bar\gamma,\bar\sigma)$; $(\E_c,\N_c)$
\State Initialize MP with $(\E_c,\N_c)$
\State $k\gets 0$; $\UB\gets +\infty$; $(\breve\gamma,\breve\sigma)\gets(\bar\gamma,\bar\sigma)$
\State $\rad_\sigma\gets 1$; $\rad_\gamma\gets 1$
\While{\textnormal{elapsed time} $< \mathrm{TL}_{\max}$}
    \State $k\gets k+1$
    \Statex \textbf{Master Problem (MP)}
    \While{MP is infeasible}
        \State $\rad_\gamma \gets \rad_\gamma + 1$
        \State Re-solve MP
    \EndWhile
    \State $\rightarrow (\gamma^k,\sigma^k)$
    \Statex \textbf{Subproblems (SPs) and Cut Generation}
    \State $\textsc{NoShed}\gets \textsc{true}$
    \ForAll{$(y,r)\in\Y\times\R$ \textbf{(in parallel)}}
        \State Fix $(\gamma_y,\sigma_y)\gets(\gamma_y^k,\sigma_y^k)$
        \State Solve Phase-I SP $\rightarrow \pi^{y,r,k}$, $\sum_{i,t}\xi_{i,y,r,t}$ 
        \If{$\sum_{i,t}\xi_{i,y,r,t} > 0$}
            \State $\textsc{NoShed}\gets \textsc{false}$; add feasibility cut \eqref{eq:TNEP:feasibility_cut}
        \Else
            \State Solve Phase-II SP $\rightarrow \pi^{y,r,k}$, $c^{\op}_{y,r}(\gamma_y^k,\sigma_y^k)$ 
            \State Add optimality cut \eqref{eq:TNEP:optimality_cut}
        \EndIf
    \EndFor

    \Statex \textbf{Bounds / TR / Level Set Update}
    \If{\textsc{NoShed}}
        \State $\hat z^k \gets c^{\mathrm{cap}}(\gamma^k,\sigma^k) + c^{\op}(\gamma^k,\sigma^k)$
        \If{$\hat z^k < \UB$}
            \State $\UB \gets \hat z^k$; $(\breve\gamma,\breve\sigma)\gets(\gamma^k,\sigma^k)$
            \State $\rad_\gamma \gets 1$
            \State Update level set \eqref{eq:TNEP:level_set}
        \EndIf
    \EndIf
    \If{$(\UB-\LB^*)/\UB \le \epsilon$}
        \State \textbf{break}
    \EndIf
\EndWhile
\State \Return $(\breve\gamma,\breve\sigma)$ with objective $\UB$
\end{algorithmic}
\end{algorithm}

\subsection{PTDF Lazy Constraint Generation}
\label{app:lazy_constraint_generation}

Algorithm \ref{alg:ptdf_lazy_constraints} outlines the PTDF lazy constraint generation used in SP solves for both Phase I and Phase II. Generated constraints are stored and persist across TR-Benders iterations.

\begin{algorithm}[h]
\caption{PTDF Lazy Constraint Generation}
\label{alg:ptdf_lazy_constraints}
\begin{algorithmic}[1]
\Require $y, r$; batch size $k_b$; fixed investments $(\gamma_y,\sigma_y)$
\State Initialize SP model with $(\gamma_y,\sigma_y)$ fixed
\Repeat
    \State Solve SP model
    \State Compute PTDF-implied flows:
    \[
        [\pf_{ij,y,r,t}]^{E}_{ij=1} \gets \Phi \times \text{inj}^{y,r,t}
    \]
    \State Compute violation magnitudes:
    \[
        v_{ij,t} \gets |\pf_{ij,y,r,t}| - |\pfmax_{ij} + \gamma_{ij,y}\Delta^c_{ij}|
    \]
    \State $\mathcal{V} \gets \{(ij,t): v_{ij,t} > 0\}$
    \If{$|\mathcal{V}| > 0$}
        \State Sort $\mathcal{V}$ by $v_{ij,t}$ (descending)
        \State $\mathcal{B} \gets$ the $k_b$ largest elements of $\mathcal{V}$
        \ForAll{$(ij,t)\in \mathcal{B}$}
            \State Add PTDF flow constraint \eqref{eq:TNEP:ptdf_flow} for $(ij,t)$:
            \[
                \pf_{ij,y,r,t} = \textstyle\sum_{k\in\N}\Phi_{ij,k}\,\text{inj}^{y,r,t}_k
            \]
            \State Add thermal limit constraint \eqref{eq:TNEP:thermal_limit} for $(ij,t)$:
            \[
                -\pfmax_{ij} - \gamma_{ij,y} \Delta^c_{ij} \leq \pf_{ij,y,r,t} \leq \pfmax_{ij} +\gamma_{ij,y} \Delta^c_{ij}
            \]
        \EndFor
    \EndIf
\Until{$|\mathcal{V}| = 0$}
\State \Return Final SP solution (and objective value)
\end{algorithmic}
\end{algorithm}

\subsection{Valid Inequalities for Minimum Required Storage Capacity}
\label{app:valid_storage_cap_lb}
For each planning period $y\in\Y$ and representative day $r\in\R$, let $s^{\mathrm{feas}}_{y,r}$ denote the minimum total storage capacity obtained by solving a feasibility-only relaxation that enforces zero load shedding, neglects operating costs, and fixes all transmission upgrades at their maximum values (i.e., $\gamma_y = \gamma^{\max}_y$).

\begin{theorem}
For any period $y\in\Y$, any investment plan that achieves zero load shedding for all $r\in\R$ must satisfy $ \sum_i \sigma_{i,y} \geq \max_{r\in\R} s^{\mathrm{feas}}_{y,r}$
\end{theorem}
\begin{proof}
    Fix $y\in\Y$. Consider any first-stage decision $(\gamma_y,\sigma_y)$ that yields zero load shedding for all $r\in\R$. Fix an arbitrary $r\in\R$. The quantity $s^{\mathrm{feas}}_{y,r}$ is obtained from a relaxation with (i) storage decisions relaxed to be continuous and (ii) transmission upgrades fixed at their maximum values, which weakly enlarges the feasible set relative to Model \ref{model:TNEP_ptdf} for day $r$. Hence, any zero-load-shed feasible plan must satisfy
    $\textstyle\sum_i \sigma_{i,y} \geq s^{\mathrm{feas}}_{y,r}$.
    Since this holds for every $r\in\R$, it follows that
    \begin{align}
    \label{eq:TNEP:valid_storage_cap_lb}
    \textstyle\sum_i \sigma_{i,y} \geq \max_{r\in\R} s^{\mathrm{feas}}_{y,r}.
    \end{align}
\end{proof}

Constraint \eqref{eq:TNEP:valid_storage_cap_lb} is imposed in the MP for each planning period $y\in\Y$ to further improve the performance of TR-Benders. The quantities $s^{\mathrm{feas}}_{y,r}$ are computed in parallel with the WS phase and then rounded to the nearest admissible discrete investment level to ensure feasibility with the discrete formulation. This yields values of $\{0.0, 27.75, 96.00, 165.00\}$ GWh for planning periods $\{2030, 2035, 2040, 2045\}$, respectively.

\subsection{Global Lower Bound for Optimality Assessment}
\label{subsec:global_lower_bound_lb}
To compute a valid global LB for the multiperiod expansion problem, costs are decomposed into per-period first-stage (investment) and second-stage (operational) components.

For each planning period $y$, a \emph{first-stage LB} is obtained by solving an investment minimization problem with continuous relaxations of the investment decisions that enforces zero load-shed, but does not consider operating costs. This problem is solved in parallel for each day $r \in \R$, yielding day-specific minimum investment costs. To ensure feasibility across all days, the maximum of these minimum investment costs across all $\R$ is taken, resulting in the minimum cumulative investment necessary to guarantee zero load-shed across all representative days in period $y$. The \emph{incremental first-stage LB} for each period is computed as the increase in the cumulative investment cost relative to the previous period.

A \emph{second-stage LB} is computed by minimizing operational costs assuming maximal feasible first-stage investments, yielding an optimistic but valid bound on operational expenditures.

Each period's total LB is formed by summing the incremental first-stage LB and the second-stage LB and applying the period discount factor. The discounted period-wise costs are summed to obtain a valid global LB for the multiperiod problem. The per-period decomposition of the LB is shown in Table \ref{tab:lower_bound_decomp}. The multiperiod problem then has $\boldsymbol{\LB^*} = \mathbf{\$28.08}\,\textbf{B}$. 

\begin{table}[h]
    \caption{Per-Period Decomposition of Valid Lower Bound (\$)}
    \label{tab:lower_bound_decomp}
    \centering
    \renewcommand{\arraystretch}{1.2}
    \begin{tabular}{|c|c|c|c|c|c|}
    \hline
    \textbf{Quantity} 
    & \textbf{2030 } 
    & \textbf{2035 } 
    & \textbf{2040 } 
    & \textbf{2045 } \\
    \hline
    Cumulative 1st-Stage LB 
        & $0.19$M 
        & $2.74$B 
        & $9.39$B 
        & $16.03$B \\
    Incremental 1st-Stage LB 
        & $0.19$M 
        & $2.74$B 
        & $6.65$B 
        & $6.65$B \\
    2nd-Stage LB 
        & $5.74$B 
        & $6.13$B 
        & $6.95$B 
        & $7.99$B \\
    \hline
    Discounted Period LB 
        & $5.74$B 
        & $6.95$B 
        & $8.35$B 
        & $7.04$B \\
    \hline
    \end{tabular}
\end{table}

\subsection{Case Study Parameters}
\label{app:case_study_parameters}

Table~\ref{tab:gen_load_scale} summarizes the projections of generation and load relative to the base year 2022, and the discount factors referenced to 2030. Table~\ref{tab:inv_cost_parameters} shows the investment cost configuration for transmission reconductoring and storage installation.

\begin{table}[h]
    \caption{Generation and Load Projections}
    \label{tab:gen_load_scale}
    \centering
    \renewcommand{\arraystretch}{1.2}
    \begin{tabular}{|c|c|c|c|c|c|}
    \hline
    \textbf{Type} & \textbf{\parbox[t]{1.8cm}{\centering Base Capacity\\(GW), 2022 \\[0.4ex]}}  &  \textbf{2030} & \textbf{2035} & \textbf{2040} & \textbf{2045} \\
    \hline
    Coal            & 14.5 & 0.82 & 0.82 & 0.82 & 0.82 \\
    Natural Gas     & 75.1 & 0.79 & 0.73 & 0.71 & 0.72 \\
    Nuclear         & 7.1  & 0.98 & 0.90 & 0.80 & 0.80 \\
    Solar           & 2.5  & 4.51 & 6.00 & 6.87 & 8.04 \\
    Wind            & 15.5 & 2.02 & 2.23 & 2.26 & 2.32 \\
    Load            & --   & 1.13 & 1.21 & 1.31 & 1.41 \\
    Discount Factor & --   & 1.00 & 0.78 & 0.61 & 0.48 \\
    \hline
    \end{tabular}
\end{table}

\begin{table}[h]
    \caption{Investment Cost Configuration}
    \label{tab:inv_cost_parameters}
    \centering
    \renewcommand{\arraystretch}{1.2}
    \begin{tabular}{|c|c|c|c|}
    \hline
    \textbf{Type} & \textbf{Overnight Cost} & \textbf{CRF} & \textbf{Annualized Cost} \\
    \hline
    Reconductoring  & \$813/MW-km       & 0.0802   & \$65/MW-km    \\
    Storage         & \$1,000,000/MWh   & 0.0963   & \$96,300/MWh  \\
    \hline
    \end{tabular}
\end{table}

\subsection{Representative days}
\label{app:representative_days}
The following 18 representative days are used in the case study, each with equal weight ($\frac{1}{18}$):
\begin{align*}
    \R = \{&
    \text{Jan 27}, \text{Feb 23}, \text{Mar 6}, \text{Mar 11}, \text{Mar 22}, \text{Mar 27}, \\
    &\text{Apr 3}, \text{Apr 22}, \text{May 10}, \text{May 19}, \text{Jun 21}, \text{Jul 11}, \\
    &\text{Aug 11}, \text{Sep 2}, \text{Sep 10}, \text{Nov 16}, \text{Dec 3}, \text{Dec 8}
    \}.
\end{align*}

\bibliographystyle{ieeetr}
\bibliography{references}

@article{ramirez2020robust_multiperiod_benders,
  title={A robust multistage approach to solve the generation and transmission expansion planning problem embedding renewable sources},
  author={Ramirez, Juan M and Hernandez, A and Marmolejo, JA},
  journal={Electric Power Systems Research},
  volume={186},
  pages={106396},
  year={2020},
  publisher={Elsevier}
}

@article{heuberger2017power_tech_cost_learning,
  title={Power capacity expansion planning considering endogenous technology cost learning},
  author={Heuberger, Clara F and Rubin, Edward S and Staffell, Iain and Shah, Nilay and Mac Dowell, Niall},
  journal={Applied Energy},
  volume={204},
  pages={831--845},
  year={2017},
  publisher={Elsevier}
}

@article{wu2025high,
  title={High-Resolution PTDF-Based Planning of Storage and Transmission Under High Renewables},
  author={Wu, Kevin and Haider, Rabab and Van Hentenryck, Pascal},
  journal={arXiv preprint arXiv:2510.14696},
  year={2025}
}

@article{levin2023energy_nature_review,
  title={Energy storage solutions to decarbonize electricity through enhanced capacity expansion modelling},
  author={Levin, Todd and Bistline, John and Sioshansi, Ramteen and Cole, Wesley J and Kwon, Jonghwan and Burger, Scott P and Crabtree, George W and Jenkins, Jesse D and O’Neil, Rebecca and Korp{\aa}s, Magnus and others},
  journal={Nature Energy},
  volume={8},
  number={11},
  pages={1199--1208},
  year={2023},
  publisher={Nature Publishing Group UK London}
}

@article{chojkiewicz2024accelerating_reconductoring_pnas,
  title={Accelerating transmission capacity expansion by using advanced conductors in existing right-of-way},
  author={Chojkiewicz, Emilia and Paliwal, Umed and Abhyankar, Nikit and Baker, Casey and O’Connell, Ric and Callaway, Duncan and Phadke, Amol},
  journal={Proceedings of the National Academy of Sciences},
  volume={121},
  number={40},
  pages={e2411207121},
  year={2024},
  publisher={National Academy of Sciences}
}

@techreport{Henriksen2013ConductorPerformance,
  author      = {Henriksen, L. and Leman, J. and Panigrahi, S.},
  title       = {Conductor Performance White Paper},
  institution = {CTC Global},
  number      = {Project No. 131411},
  year        = {2013},
  url         = {https://ctcglobal.com/wp-content/uploads/2023/05/ACCC_Conductor_Performance_by_Power_Engineers.pdf},
  note        = {[Online]. Available}
}

@misc{EIA2023Outlook,
  author       = {{U.S. Energy Information Administration}},
  title        = {{Annual Energy Outlook 2023}},
  year         = 2023,
  url          = {https://www.eia.gov/outlooks/aeo/},
  note         = {Accessed: 2024-11-23}
}

@article{synthetic_1,
  title={Grid structural characteristics as validation criteria for synthetic networks},
  author={Birchfield, Adam B and Xu, Ti and Gegner, Kathleen M and Shetye, Komal S and Overbye, Thomas J},
  journal={IEEE Transactions on power systems},
  volume={32},
  number={4},
  pages={3258--3265},
  year={2016},
  publisher={IEEE}
}

@techreport{NRELBatteryCost,
  title        = {Cost Projections for Utility-Scale Battery Storage: 2023 Update},
  author       = {Wesley Cole and Akash Karmakar},
  institution  = {National Renewable Energy Laboratory},
  year         = {2023},
  month        = {June},
  number       = {NREL/TP-6A40-85332},
  url          = {https://www.nrel.gov/docs/fy23osti/85332.pdf}
}

@misc{accc_longevity,
  author = {{CTC Global Corporation}},
  title = {{ACCC\textregistered\ Conductor Longevity – A Brief Overview}},
  year = {2024},
  month = jun,
  note = {Accessed: 2026-01-14},
  url = {https://ctcglobal.com/accc-conductor-longevity-a-brief-overview/}}

@misc{gurobi,
  author = {{Gurobi Optimization, LLC}},
  title = {{Gurobi Optimizer Reference Manual}},
  year = 2023,
  url = "https://www.gurobi.com"
}

@manual{PACE,
     title  = "{P}artnership for an {A}dvanced {C}omputing {E}nvironment ({PACE})",
     author = "PACE",
     url    = "http://www.pace.gatech.edu",
     year   = "2017"
}

@article{Lubin2023_jump,
    author = {Miles Lubin and Oscar Dowson and Joaquim {Dias Garcia} and Joey Huchette and Beno{\^i}t Legat and Juan Pablo Vielma},
    title = {{JuMP} 1.0: {R}ecent improvements to a modeling language for mathematical optimization},
    journal = {Mathematical Programming Computation},
    year = {2023},
    doi = {10.1007/s12532-023-00239-3}
}

@article{goke2024stabilized,
  title={Stabilized Benders decomposition for energy planning under climate uncertainty},
  author={G{\"o}ke, Leonard and Schmidt, Felix and Kendziorski, Mario},
  journal={European Journal of Operational Research},
  volume={316},
  number={1},
  pages={183--199},
  year={2024},
  publisher={Elsevier}
}

@article{zuluaga2024cats_progressive_hedging_ph,
  title={Parallel computing for power system climate resiliency: Solving a large-scale stochastic capacity expansion problem with mpi-sppy},
  author={Zuluaga, Tomas Valencia and Musselman, Amelia and Watson, Jean-Paul and Oren, Shmuel S},
  journal={Electric Power Systems Research},
  volume={235},
  pages={110720},
  year={2024},
  publisher={Elsevier}
}

@article{piansky2025optimizing_cats_wildfire_benders,
  title={Optimizing battery and line undergrounding investments for transmission systems under wildfire risk scenarios: A benders decomposition approach},
  author={Piansky, Ryan and Gupta, Rahul and Molzahn, Daniel K},
  journal={Sustainable Energy, Grids and Networks},
  pages={101838},
  year={2025},
  publisher={Elsevier}
}

@article{rahmani2016comprehensive,
  title={Comprehensive power transfer distribution factor model for large-scale transmission expansion planning},
  author={Rahmani, Mohsen and Kargarian, Amin and Hug, Gabriela},
  journal={IET Generation, Transmission \& Distribution},
  volume={10},
  number={12},
  pages={2981--2989},
  year={2016},
  publisher={Wiley Online Library}
}

@techreport{ReEDS,
  title={Regional Energy Deployment System (ReEDS) Model Documentation (Version 2020)},
  author={Ho, Jonathan and Becker, Jonathon and Brown, Maxwell and Brown, Patrick and Chernyakhovskiy, Ilya and Cohen, Stuart and Cole, Wesley and Corcoran, Sean and Eurek, Kelly and Frazier, Will and others},
  year={2021},
  institution={National Renewable Energy Lab.(NREL), Golden, CO (United States)}
}

@article{Jacobson,
  title={Quantifying the Impact of Energy System Model Resolution on Siting, Cost, Reliability, and Emissions},
  author={Jacobson, Anna F and Mauzerall, Denise L and Jenkins, Jesse D},
  journal={arXiv preprint arXiv:2406.16924},
  year={2024}
}

@article{XuHobbs,
  title={Value of model enhancements: quantifying the benefit of improved transmission planning models},
  author={Xu, Qingyu and Hobbs, Benjamin F},
  journal={IET Generation, Transmission \& Distribution},
  volume={13},
  number={13},
  pages={2836--2845},
  year={2019},
  publisher={Wiley Online Library}
}

@article{pecci2025regularized,
  title={Regularized benders decomposition for high performance capacity expansion models},
  author={Pecci, Filippo and Jenkins, Jesse D},
  journal={IEEE Transactions on Power Systems},
  volume={40},
  number={4},
  pages={3105--3116},
  year={2025},
  publisher={IEEE}
}
\balance

\endgroup
\end{document}